\documentclass[conference]{IEEEtran}

\usepackage[english]{babel}

\usepackage{nicefrac,mathtools,units,xspace}
\usepackage{amsmath,amsfonts,amssymb,amsthm,mathtools}
\usepackage[caption=false,font=footnotesize]{subfig}
\usepackage{graphicx,booktabs}
\usepackage[boxed]{algorithm}
\usepackage{algorithmic,afterpage,comment}
\usepackage[disable]{todonotes}
\graphicspath{{./}{./fig/}}

\usepackage{tikz}
\usepackage{pgfplots}
\usepackage{tikz-3dplot}
\usetikzlibrary{arrows,shapes.geometric,calc,decorations.markings,patterns}
\definecolor{mathblue}{HTML}{3F3D9A}
\definecolor{mathpurp}{HTML}{9A3D71}
\definecolor{mathsand}{HTML}{9A8C3D}
\definecolor{mathgrn}{HTML}{0F9F4F}

\usepackage[compress]{cite}

\newtheorem{proposition}{Proposition}

\newtheorem{lemma}{Lemma}
\newtheorem{problem}{Problem}

\def\vzer{{\bs 0}}

\def\bs{\boldsymbol}
\def\bb{\mathbb}
\def\cl{\mathcal}
\def\sf{\mathsf}
\def\ts{\textstyle}

\newcommand{\gone}{\sqrt{\scalebox{.8}{$\frac{\pi}{2}$}}}

\newcommand{\tinv}[1]{{\textstyle\frac{1}{#1}}}

\renewcommand{\leq}{\leqslant}
\renewcommand{\geq}{\geqslant}

\newcommand{\tiid}{\emph{\tiid}\xspace}
\newcommand{\rv}{\emph{r.v.}\xspace}
\newcommand{\rvs}{\emph{r.v.'s}\xspace}
\newcommand{\whp}{\emph{w.h.p.}\xspace}

\DeclareMathOperator{\rad}{rad}

\newcommand{\ie}{\emph{i.e.}, }
\newcommand{\eg}{\emph{e.g.}, }

\newcommand{\vtodo}[1]{\todo[inline]{#1}}

\def\eg{\emph{e.g.},~}
\def\etal{{\emph{et al.}~}}
\def\ie{\emph{i.e.},~}

\makeatletter
\def\thm@space@setup{\thm@preskip=1ex
	\thm@postskip=1ex}
\makeatother
\begin{document}
\title{\vspace{-3mm}The Rare Eclipse Problem on Tiles: Quantised Embeddings of Disjoint Convex Sets\vspace{-2mm}}

\IEEEoverridecommandlockouts

\author{\IEEEauthorblockN{Valerio Cambareri, Chunlei Xu, Laurent Jacques}
	\IEEEauthorblockA{ISPGroup, ICTEAM/ELEN, Universit\'e catholique de Louvain, Louvain-la-Neuve, Belgium.
		\thanks{E-mail: {\{valerio.cambareri,~chunlei.xu,~laurent.jacques\}@uclouvain.be}. The authors are partly funded by the Belgian National Fund for Scientific Research (FNRS) under the M.I.S.-FNRS project {\sc AlterSense}. All authors have equally contributed to the realisation of this paper.}
	\vspace{-2mm}}
}

\maketitle
\begin{abstract}
Quantised random embeddings are an efficient dimensionality reduction technique which preserves the distances of low-complexity signals up to some controllable  additive and multiplicative distortions. In this work, we instead focus on verifying when this technique preserves the separability of two disjoint closed  convex sets, \ie in a quantised view of the ``rare eclipse problem'' introduced by Bandeira \etal in 2014. This separability would ensure exact classification of signals in such sets from the signatures output by this non-linear dimensionality reduction. We here present a result relating the embedding's dimension, its quantiser resolution and the sets' separation, as well as some numerically testable conditions to illustrate it. Experimental evidence is then provided in the special case of two $\ell_2$-balls, tracing the phase transition curves that ensure these sets' separability in the embedded domain.
\end{abstract}

	
\section{Introduction}
\label{sec:introduction}
Dimensionality reduction methods are a crucial part of very large-scale machine learning frameworks, as they are in charge of mapping (with negligible losses) the information contained in high-dimensional data to a low-dimensional domain, thus minimising the computational effort of learning tasks. We here focus on a class of {\em non-linear}, {\em non-adaptive} dimensionality reduction methods, \ie {\em quantised random embeddings}, as obtained\footnote{Our notation conventions are reported at the end of this section.} by applying to $\bs x \in \cl K$ (with $\cl K \subset \bb R^n$ any {\em dataset}) 
\begin{equation}
\label{eq:nle}
\bs y = {\sf A} (\bs x) \coloneqq \cl Q_{\delta}(\bs\Phi \bs x + \bs\xi), 
\end{equation}
where $\bs\Phi \in \bb R^{m\times n}$ is a Gaussian random~{\em sensing matrix}, \ie $\bs\Phi \sim \cl N^{m \times n}(0,1)$; $\cl Q_{\delta} (\cdot) \coloneqq \delta \lfloor \frac{\cdot}{\delta} \rfloor$ is a {\em uniform scalar quantiser} of {\em resolution} $\delta>0$ (applied component-wise), yielding a {\em signature} $\bs y \in \delta \bb Z^m$; $\bs\xi  \sim {\cl U}^{m}([0,{\delta}])$ is some {\em dither} drawn uniformly in $[0,\delta]^m$, which is fundamental to stabilise the action of the quantiser~\cite{GrayNeuhoff1998,Jacques2015}. 

The non-linear map described by \eqref{eq:nle} produces \emph{compact}~signatures $\bs y$, either in terms of dimension $m \ll n$, or of bits per entry (controlled by $\delta$) even if $m > n$~\cite{BoufounosJacquesKrahmerEtAl2015}. 
Learning tasks such as classification may then run on $\bs y \in \sf A(\cl K)$ rather than $\bs x \in \cl K$ at reduced storage, transmission, and computational costs, with accuracy depending on $m$, $\delta$. However, contrarily to other non-linear maps  (\eg\cite{RahimiRecht2008,BoufounosRaneMansour2015}), \eqref{eq:nle} retains {\em quasi-isometry} properties~\cite{Jacques2015, JacquesCambareri2016} that grant, under some requirements on $m$ (\ie {\em sample complexity} bounds), the \emph{recovery} of $\bs x$ from $\bs y$ using appropriate algorithms~\cite{MoshtaghpourJacquesCambareriEtAl2016}. 

In this contribution we aim to prove that generic learning tasks can run seamlessly on $\sf A(\cl K)$ by the ability of~\eqref{eq:nle} to preserve the \emph{separability} of different \emph{classes} in $\cl K$. These classes are described by $Q$ disjoint closed convex sets, \ie  $\cl C_i \subset \cl K, {i\in [Q]} $ so that $\forall i,j \in [Q], i\neq j, \cl C_i\cap\cl C_j=\emptyset$. Hence, we inquire whether testing if our data $\bs x \in \cl C_i$ is \emph{equivalent} to doing so given $\bs y$ in \eqref{eq:nle}; for this to hold, it is necessary that the classes' images $\sf A(\cl C_i), {i\in [Q]}$ are still \emph{separable}, \ie $\forall i,j \in [Q], i\neq j, \, \sf A(\cl C_i)\cap \sf A(\cl C_j) = \emptyset$. If this is violated then no learning algorithm can perform exact classification, as the images would ``eclipse'' each other. 
This perspective builds upon that of Bandeira \etal\cite{BandeiraMixonRecht2014}, who introduced this \emph{rare eclipse problem} for linear embeddings, as reviewed in Sec.~\ref{sec:the-rare-eclipse-problem}.  
Focusing on $Q=2$ classes, in Sec.~\ref{sec:the-quantised-eclipse-problem} we define the \emph{quantised eclipse problem} and present our main result, \ie a sample complexity bound which states the conditions on $m$, $\delta$, $\bs\Phi$, and $\cl C_i, i \in [Q]$ under which the images $\sf A(\cl C_i)$ are separable {with high probability} (\whp). In Sec.~\ref{sec:testable-conditions-by-convex-optimisation} this is simplified by lower bounds to the latter probability, which have the advantage of being numerically testable for disjoint convex sets by solving convex optimisation problems. Among such sets, we detail the specific case of two high-dimensional~$\ell_2$-balls in Sec.~\ref{sec:the-case-of-two-euclidean-balls}; this is explored numerically in Sec.~\ref{sec:numerical-experiments} by computing phase transition curves on the above probability bound, indicating a regime with respect to (\emph{w.r.t.})~$m,\delta$ for \eqref{eq:nle} in which the sets' separability is preserved.

{\em Notation:}   
Given a random variable (\rv) $X$ (\eg normal $\cl N(0,1)$ or uniform $\cl U([0,\delta])$ \rvs),
we write $\bs U \sim X^{d_1 \times d_2}$ (\eg $\cl N^{d_1\times d_2}(0,1)$) to denote the $d_1 \times d_2$ matrix (or vector, if $d_2=1$) with {independent and identically distributed} ({\em i.i.d.}) entries $U_{ij} \sim_{\rm i.i.d.} X$. 
Spheres and balls in $\ell_p(\bb R^q)$ are denoted by $\bb S^{q-1}_p$ and $\bb B^q_p$. For a set $\cl C \subset \bb R^n$, its Chebyshev radius is $\rad(\cl C)=\inf\{r>0: \exists \bs c \in \bb R^n,\ \cl C \subset \bs c + r \bb B^n\}$; its image under a map $\sf A$ is $\sf A(\cl C)$;  its projection by a matrix $\bs B$ is $\bs B \cl C$. The cardinality of a set $\cl C$ reads $|\cl C|$, and $[Q] \coloneqq \{1,\ldots,Q\}$. We denote by $C,c,c',c''$ constants whose value can change between
lines. We also write $f\lesssim g$ if $\exists c>0$ such that $f \leq c\,g$, and correspondingly for $f\gtrsim g$. Moreover, $f\simeq g$ means that $f\lesssim g$ and $g\lesssim f$.

{\em Relation to Prior Work:}
Many contributions have discussed linear dimensionality reduction by $\bs y = \bs\Phi \bs x$ with $\bs\Phi\sim X^{m \times n}$ a random matrix having {\em i.i.d.} entries distributed as a sub-Gaussian \rv $X$ (for a survey, see~\cite{Vershynin2012}), \ie \emph{random projections}. Following the work of Johnson and Lindenstrauss~\cite{JohnsonLindenstrauss1984}, such linear embeddings were soon recognised~\cite{DasguptaGupta2003,Achlioptas2003} as distance-preserving, non-adaptive\footnote{Not requiring any potentially large or unavailable training dataset, as opposed to,~\eg principal component analysis} dimensionality reductions for \emph{finite datasets},~\ie with $|\cl K| < \infty$. 
Moreover, several non-linear random embeddings are now available for more general models of $\cl K$~\cite{BoufounosRaneMansour2015,OymakRecht2015,Jacques2015a,Jacques2015,PlanVershynin2014,PlanVershynin2013}; most results on such embeddings rely on preserving distances, rather than the separation between classes within $\cl K$. Regarding this last aspect, Dasgupta~\cite{Dasgupta1999} first analysed the separability of a mixture-of-Gaussians dataset $\cl K$ after random projections. Later, with the rise of {Compressed Sensing} (CS), random projections followed by classification tasks were dubbed \emph{compressive classification}. Davenport \etal\cite{DavenportBoufounosWakinEtAl2010} showed that if $\bs\Phi$ verifies the Restricted Isometry Property (RIP) \emph{w.r.t.} a dataset $\cl K$ (\ie a \emph{stable embedding}) then exact classification can be achieved on $\bs\Phi \cl K$ thanks to distance preservation; $\cl K$ was therein taken as a finite set, or the set of sparse signals. Reboredo \etal\cite{ReboredoRennaCalderbankEtAl2013,ReboredoRennaCalderbankEtAl2016} studied the limits of compressive classification in a Bayesian framework. Finally, Bandeira \etal\cite{BandeiraMixonRecht2014} first explored with the tools of high-dimensional geometry the conditions for the separability of closed convex sets $\cl C_1, \cl C_2 \subset \cl K \coloneqq \bb R^n$ after random projections. We here extend their approach to quantised random embeddings given by \eqref{eq:nle} which, due to their non-linearity, is a non-trivial endeavour that is currently lacking in the  literature.
\section{Quantised Random Embeddings and \\ the Rare Eclipse Problem}
\subsection{The Rare Eclipse Problem}\label{sec:the-rare-eclipse-problem} 
Let us first recall the fundamental question introduced by Bandeira \etal\cite{BandeiraMixonRecht2014} and their main result as follows. 
\begin{problem}[Rare Eclipse Problem (from~\cite{BandeiraMixonRecht2014})]
\label{prob:rep}
Let~${\cl C}_1,{\cl C}_2 \subset {\bb R}^n:$~${\cl C}_1\cap{\cl C}_2=\emptyset$~be closed convex sets,~$\bs\Phi \mathop{\sim} \cl N^{m\times n}(0, 1)$. Given~$\eta \in (0, 1)$, find the smallest $m$ so that 
\begin{equation}
\label{eq:rep}
p_{0} \coloneqq \bb P[\bs\Phi \cl C_1 \cap \bs\Phi \cl C_2 = \emptyset] \geq 1-\eta.
\end{equation}
\end{problem}
\noindent Prob.~\ref{prob:rep} is equivalent to ensuring, for all $\bs x_1 \in \cl C_1$, $\bs x_2 \in \cl C_2$, that $\bs\Phi \bs x_1 \neq \bs\Phi \bs x_2$ with $\bs\Phi \sim \cl N^{m \times n}(0,1)$. Let us define the {\em difference set}~$\cl C^- \coloneqq \cl C_1 - \cl C_2= \{\bs z \coloneqq \bs x_1-\bs x_2 : \bs x_1 \in \cl C_1, \bs x_2 \in \cl C_2\}$. We can then cast \eqref{eq:rep} in terms of the kernel of $\bs\Phi$, \ie
\begin{IEEEeqnarray}{rCl}
p_0 & = &  \ts \bb P[\forall \bs z \in \cl C^-, \bs\Phi \bs z \neq \vzer _m] = 1 -  \bb P[\exists \bs z \in \cl C^-: \bs\Phi \bs z = \vzer_m] \nonumber \\ & \geq  & 1 - \eta, \, \textit{i.e.}, \, \bb P[{\rm Ker}(\bs\Phi) \cap \cl C^- \neq \emptyset] \leq \eta \label{eq:probintersect}.
\end{IEEEeqnarray}
Intuitively, $\eta$ in \eqref{eq:rep} will increase with the ``size'' of $\cl C^-$, as its intersection with ${\rm Ker}(\bs\Phi)$ will be non-empty. This size is here measured by the {\em Gaussian mean width}, \ie for any set $\cl C$,
\[
w(\cl C)\coloneqq\bb E_{\bs g} \sup_{\bs x \in \cl C}|\bs g^\top \bs x|, \, \bs g \sim \cl N^{n}(0, 1).
\]
Bandeira \etal then realised that \eqref{eq:probintersect} is found by Gordon's Escape Theorem \cite{Gordon1988} since, by arbitrarily scaling $\cl C^-$ that amounts to taking the cone $\bb R_+ \cl C^-$, and by its intersection with the sphere $\bb S^{n-1}_2$, we obtain a \emph{mesh} (\ie a closed subset of $\bb S^{n-1}_2$). Let us then define $\cl S \coloneqq (\bb R_+ \cl C^-) \cap \bb S^{n-1}_2$ of width $w_\cap \coloneqq w(\cl S)$, and report their main result (its proof is in \cite{BandeiraMixonRecht2014}).
\begin{proposition}[Corollary 3.1 in \cite{BandeiraMixonRecht2014}]
	\label{prop:mainrep}
	In the setup of Prob.~\ref{prob:rep}, given $\eta \in (0,1)$, if
	$m \gtrsim (w_\cap  + \sqrt{2 \log \tfrac{1}{\eta}})^2 + 1$ 
	then $p_0 \geq 1 - \eta$.
\end{proposition}
\noindent Hence, the sample complexity of Prob.~\ref{prob:rep} is sharply characterised for any difference set whose $w_\cap$ is given or bounded.
%
%
\subsection{The Quantised Eclipse Problem}
\label{sec:the-quantised-eclipse-problem}
Extending Prop.~\ref{prop:mainrep}~to quantised random embeddings as in~\eqref{eq:nle} is not simple. To begin with, any two closed convex sets $\cl C_1, \cl C_2$ would now be mapped into two \emph{countable} sets $\sf A(\cl C_1), \sf A(\cl C_2) \subset \delta \bb Z^m$; verifying when they ``collide'' is our key question below.
\begin{problem}[Quantised Eclipse Problem]
\label{prob:qrep}
Let~${\cl C}_1,{\cl C}_2 \subset {\bb R}^n:$~${\cl C}_1\cap{\cl C}_2=\emptyset$~be closed convex sets, and $\sf A$ defined in \eqref{eq:nle} with $\delta > 0$. Given~$\eta \in (0, 1)$, find the smallest $m$ so that
\begin{equation}
\label{eq:pdelta}
p_\delta \coloneqq \bb P[\sf A(\cl C_1) \cap \sf A(\cl C_2) = \emptyset] \geq 1-\eta.
\end{equation}
\end{problem}
\noindent Note that, since $\sf A$ itself uses $\bs\Phi \mathop{\sim} \cl N^{m\times n}(0, 1)$ before quantisation, $\bs\Phi \cl C_1 \cap \bs\Phi \cl C_2 \neq \emptyset \implies \sf A(\cl C_1) \cap \sf A(\cl C_2) \neq \emptyset$; hence, $p_0 \geq p_\delta$ given the same $\bs\Phi, \cl C_1, \cl C_2$. However, the converse does not hold since $\Phi \cl C_1 \cap \Phi \cl C_2 = \emptyset$ by itself does not suffice to ensure $\sf A(\cl C_1) \cap \sf A(\cl C_2) = \emptyset$ due to,~\eg coarse quantisation with large $\delta$ or some draws of $\bs\xi$ in \eqref{eq:nle}. Then, letting the event $\sf E \coloneqq \{\forall \bs x_1 \in \cl C_1, \bs x_2 \in \cl C_2, \sf A(\bs x_1) \neq \sf A(\bs x_2)\}$, we see \eqref{eq:pdelta} equals
\begin{IEEEeqnarray}{C}
	\label{eq:secondformpdelta} p_\delta = \bb P[\sf E] \geq 1-\eta, \, \textit{i.e.}, \, \bb P[\sf E^{\rm c}]\leq \eta,\\
	\sf E^{\rm c}\coloneqq\{\exists~\bs x_1 \in \cl C_1,\bs x_2 \in \cl C_2 : \sf A(\bs x_1) = \sf A(\bs x_2)\}
\end{IEEEeqnarray}
Hence, $\eta$ bounds the probability that any two $\bs x_1 \in \cl C_1$, $\bs x_2 \in \cl C_2$ are \emph{consistent}. Note that, by consistency \cite{Jacques2015}, $\sf A(\bs x_1) = \sf A(\bs x_2) \implies \|\bs\Phi \bs z \|_\infty < \delta$ with $\bs z = \bs x_1 - \bs x_2 \in \cl C^-$. Thus, introducing the {\em separation} $\sigma \coloneqq \min_{\bs z \in \cl C^-} \|\bs z\|_2$, it is expected that $\eta$ will decay to $0$ as $\sigma$ increases and $\delta$ decreases. 

This is also sustained by the fact that $\sf A$ is known to respect \whp the Quantised Restricted Isometry Property (QRIP)~\cite{JacquesCambareri2016} over some $\cl K \subset \bb R^n$ provided  $\bs \Phi$ satisfies a $(\ell_1,\ell_2)$-form of the RIP (see Lemma~\ref{lemma:rip12}) and $m$ is large before the dimension of $\cl K$. If the QRIP holds, we would then have, for all $\bs u_1,\bs u_2 \in \cl K$, 
$$
\ts |\frac{1}{m}\|\sf A(\bs u_1)-\sf A(\bs u_2)\|_1 - c' \|\bs u_1 - \bs u_2\|_2|\leq \varepsilon \|\bs u_1 - \bs u_2\|_2 + c \delta \varepsilon', 
$$
for some controllable distortions $\varepsilon,\varepsilon' >0$ and constants $c,c'>0$. With $\cl K \coloneqq \cl C_1 \cup \cl C_2$, $\bs x_1 \in \cl C_1$ and $\bs x_2 \in \cl C_2$, this ensures that $\frac{1}{m}\|\sf A(\bs x_1)-\sf A(\bs x_2)\|_1 \geq (c'-\varepsilon) \|\bs z\|_2 - c \delta \varepsilon' \geq (c'-\varepsilon) \sigma - c \delta \varepsilon'$. Thus, $\sf A(\bs x_1) \neq \sf A(\bs x_2)$ simply follows if $\tfrac{\sigma}{\delta} > \tfrac{c \varepsilon'}{c' - \varepsilon}$. 

Before introducing our main result, let us present two lemmata, whose proof is given in the Appendix. 
The first assesses when $\bs\Phi \sim \cl N^{m \times n}(0,1)$ respects a $(\ell_1, \ell_2)$-form of the RIP for a {mesh} (see,~\eg\cite[Cor. 2.3]{PlanVershynin2014},\cite{Schechtman2006}).
\begin{lemma}
	\label{lemma:rip12}
	Let $\epsilon_0 >0$ and $\cl S \subset \bb S_2^{n-1}$. If $m \gtrsim \epsilon_0^{-2} w^2(\cl S)$ and $\bs\Phi \sim \cl N^{m\times n}(0,1)$, then there exist some $C,c>0$ such that, with probability exceeding $1-C\exp(-c \epsilon_0^2 m)$ and $\kappa_0 = \gone$, 
	\begin{equation}
	\label{eq:rip12}
	\ts  (1-\epsilon_0) \leq \frac{\kappa_0}{m} \|\bs\Phi \bs u\|_1 \leq (1+\epsilon_0),\quad\forall \bs u \in \cl S.
	\end{equation}
\end{lemma}
\noindent Thus, provided $m \gtrsim \epsilon_0^{-2} w^2_\cap$ and defining $\cl D_{\ell_1}(\bs a,\bs b) \coloneqq \tinv{m} \|\bs a -\bs b\|_1$, applying Lemma~\ref{lemma:rip12} to $\cl S \coloneqq (\bb R_+ \cl C^-) \cap \bb S^{n-1}_2$ yields
\begin{equation}
\label{eq:rip12-Cm}
\ts \big|\kappa_0 \cl D_{\ell_1}(\bs\Phi \bs x_1\!, \bs\Phi \bs x_2)  - \|\bs x_1 - \bs x_2\|_2 \big|\ \leq\ \epsilon_0 \|\bs x_1 - \bs x_2\|_2,
\end{equation}
with the same probability and for all $\bs x_1 \in \cl C_1, \bs x_2 \in \cl C_2$, since $\tfrac{\bs x_1 - \bs x_2}{\|\bs x_1 - \bs x_2\|_2} \in \cl S$. 
Moreover, since $w^2(\bb S_2^{n-1}) \lesssim n$, provided $m
\gtrsim \epsilon^{-2} n$ for some $\epsilon >0$, we also have with probability exceeding $1- C\exp(-c \epsilon^2 m)$,  
\begin{equation}
\label{eq:rip12-Rn}
\ts  (1-\epsilon) \|\bs u\|_2 \leq \frac{\kappa_0}{{m}} \|\bs\Phi \bs u\|_1 \leq (1+\epsilon) \|\bs u\|_2,\ \forall \bs u \in \bb R^n.
\end{equation}

The second lemma proves that the mapping ${\sf A}'(\cdot) \coloneqq \cl Q(\cdot
+ \bs \xi)$, with $\bs \xi \sim \cl U^m([0,\delta])$, embeds\footnote{That is, in the Gromov-Hausdorff sense~\cite{PlanVershynin2014}.} \whp $\bb R^m$ in $\delta \bb Z^m$
in the metric $\cl D_{\ell_1}$ and up to some controlled distortions. This lemma uses the Kolmogorov
entropy $\cl H_q(\cl E, \rho) \coloneqq \log \cl N_q(\cl E, \rho)$ of a bounded subset $\cl E \subset \bb R^m$
in the $\ell_q$-metric ($q\geq 1$) defined for $\rho>0$, with $\cl N_p(\cl E, \rho)$ the cardinality of its smallest $\rho$-covering in the same metric.  
\begin{lemma}
	\label{lemma:tfd-dithered-embed} Let $\cl E \subset \bb R^m$ be a bounded set. Given $\epsilon, \delta > 0$, if
	$$
	\ts m \gtrsim \epsilon^{-2} \cl H_1(\cl E, \tfrac{m\delta\epsilon^2}{1+\epsilon}),
	$$
	then, for $\bs \xi \sim \cl U^m([0,\delta])$ and with probability exceeding $1 - C\exp(-cm\epsilon^2)$ for some $C,c>0$, we have 
	\begin{equation}
	\label{eq:1}
	\ts \big| \cl D_{\ell_1}({\sf A}'(\bs a),{\sf A}'(\bs b)) - \cl D_{\ell_1}(\bs a,\bs b)\big| \lesssim \delta\epsilon,\ \forall \bs a,\bs b \in \cl E.
	\end{equation}
\end{lemma}

\noindent We are finally able to state our main result, solving Prob.~\ref{prob:qrep}.
\begin{proposition}
	\label{prop:main}
	In the setup of Prob.~\ref{prob:qrep}, let $r_i \coloneqq \rad(\cl C_i)$, $i \in \{1,2\}$, $r \coloneqq r_1 + r_2$, and ${\sf A}$ defined in \eqref{eq:nle} with $\delta > 0$. Given $\eta \in (0,1)$, if
	\begin{equation}
	\label{eq:cond-m}
	\ts m \gtrsim (w^2_\cap  + 
	n\frac{\delta^2}{\sigma^2}) (1 + \log (1 +
	\frac{r m}{\delta n}) + w^{-2}_\cap \log\tinv{\eta}), 
	\end{equation}
	then  $p_\delta \geq 1 - \eta$.
\end{proposition}
\begin{proof}[Proof of Prop.~\ref{prop:main}]
	Let us first observe when \eqref{eq:1} holds with $\bs\Phi \sim \cl N^{m \times n}(0,1)$, $\cl E \coloneqq \bs\Phi \cl C_\cup$, $\cl C_\cup \coloneqq \cl C_1 \cup \cl C_2 \subset \bb R^n$. This will be useful later to characterise when ${\sf A}(\cl C_1) \cap {\sf A}(\cl C_2) = \emptyset$. Let $\cl R_\cup$ be a $\rho$-covering in the $\ell_2$-metric
	of $\cl C_\cup$ for some $\rho>0$ to be specified below. If $m \gtrsim \epsilon^{-2} n$ for some $\epsilon > 0$, we have from
	\eqref{eq:rip12-Rn} that, with probability exceeding
	$1-C\exp(-c\epsilon^2 m)$,
	the event ${\sf E}_0$ where $\bs\Phi \cl R_\cup$ is a
	$\rho'$-covering of $\bs\Phi \cl C_\cup$ holds with $\rho' =
	2m(1+\epsilon)\rho$.
	This proves that, conditionally to ${\sf E_0}$ and for $\cl E = \bs\Phi \cl C_\cup$, $\cl H_1(\cl E,
	\rho') \leq \cl H_2(\cl C_\cup, \rho)$.
	However, $\cl H_2(\cl C_\cup, \rho) \leq \log 2 + \max(\cl H_2(\cl C_1,
	\rho), \cl H_2(\cl C_2, \rho)) \lesssim \max(\cl H_2(\cl C_1,
	\rho), \cl H_2(\cl C_2, \rho))$. Moreover, we have $\cl H_2(\cl C_i,
	\rho) \lesssim n \log (1 + \frac{r_i}{\rho})$~\cite{Pisier1999}, so that $\cl H_2(\cl
	C_\cup, \rho) \lesssim n \log (1 + \frac{r}{\rho})$. Setting $\rho' \coloneqq
	\tfrac{m\delta\epsilon^2}{1+\epsilon}$ gives $\rho =
	\tfrac{\delta\epsilon^2}{2(1+\epsilon)^2}$ and finally 
	$$
	\ts \cl H_1(\bs\Phi \cl C_\cup, \rho) \lesssim n \log (1 + \frac{2r(1+\epsilon)^2}{\delta\epsilon^2}).
	$$
	Consequently, conditionally to ${\sf E}_0$ which only depends on $\bs
	\Phi$, Lemma~\ref{lemma:tfd-dithered-embed} provides that if $m \gtrsim
	\epsilon^{-2} n \log (1 + \frac{2r(1+\epsilon)^2}{\delta\epsilon^2})$
	then, with probability exceeding $1-C\exp(-cm\epsilon^2)$,
	we get the occurrence of a new event, ${\sf E}'_0$, where \eqref{eq:1} holds with $\bs a = \bs\Phi \bs u$ and $\bs b = \bs\Phi
	\bs v$ for all $\bs u, \bs v \in \cl C_{\cup}$.  Under the same conditions, since $\bb
	P[{\sf E}'_0] \geq \bb P[{\sf E}'_0|{\sf E}_0] \bb P[{\sf E}_0]$, ${\sf E}'_0$ occurs unconditionally with $\bb P[{\sf E}'_0] \geq 1-C'\exp(-c'm\epsilon^2)$, for some $C',c'>0$.

	Second, if $m \gtrsim \epsilon_0^{-2} w^2_\cap$ for some $\epsilon_0 > 0$, Lemma~\ref{lemma:rip12} states that the event ${\sf
		E}_1$, where \eqref{eq:rip12-Cm} is respected for all $\bs x_1 \in \cl
	C_1$ and all $\bs x_2 \in \cl
	C_2$, holds with probability exceeding  $1-C\exp(-c\epsilon_0^2 m)$. 
	
	Given $\eta\!\in\!(0,1)$ and $\epsilon\!=\!\epsilon(\epsilon_0)\!\coloneqq\!\frac{\sqrt n}{w_\cap}\!\epsilon_0$, \ie with $\epsilon
	\gtrsim\!\epsilon_0$ since $w_\cap\!\lesssim\!\sqrt n$, the union
	bound yields that ${\sf E_0}$ and ${\sf E_1}$ jointly hold with
	probability exceeding $1-C\exp(-c\epsilon_0^2\!m)\!\geq\!1\!-\!\eta$ provided                
	\begin{equation}
	\label{eq:cond-on-m}
	\ts m \gtrsim
	\epsilon^{-2}_0 \big(w^2_\cap\,(1 +\log (1 +
	\frac{2r\,(1+\epsilon(\epsilon_0)\,)^2}{\delta\,\epsilon(\epsilon_0)^2})) + \log\tinv{\eta}\big). 
	\end{equation}
	
	In this case, for all $\bs u \in \cl C_1$ and $\bs v \in \cl C_2$ (or \emph{vice versa}), \eqref{eq:1} (with $\bs a \coloneqq \bs\Phi \bs u$ and $\bs b \coloneqq \bs\Phi
	\bs v$) 
	and \eqref{eq:rip12-Cm}
	give, for some $c>0$,
	\begin{align*}
	&\ts \cl D_{\ell_1}({\sf A}(\bs u),{\sf A}(\bs v)) = \cl
	D_{\ell_1}({\sf A}'(\bs\Phi \bs u),{\sf A}'(\bs\Phi \bs v))\\
	&\ts \geq \cl
	D_{\ell_1}(\bs\Phi \bs u,\bs\Phi \bs v) - c \delta\epsilon \geq
	\kappa_0^{-1}\,(1-\epsilon_0) \|\bs u - \bs v\|_2 - c \delta\epsilon\\
	&\ts \geq \kappa_0^{-1}\,(1-\epsilon_0) \sigma - c \delta\epsilon = \kappa_0^{-1}\,(1-\epsilon_0) \sigma - \frac{c\delta\sqrt n}{w_\cap} \epsilon_0.  
	\end{align*}
	In order to have ${\sf A}(\cl C_1) \cap {\sf A}(\cl C_2) = \emptyset$, the last
	quantity must be positive. Since $\epsilon_0 > 0$, this clearly happens if
	$\kappa_0^{-1}\,(1-\epsilon_0) \sigma - \frac{c\delta\sqrt
		n}{w_\cap} \epsilon_0 = \frac{c\delta\sqrt
		n}{w_\cap} \epsilon_0$, which gives 
	$$
	\ts \epsilon^{-2}_0 w^2_\cap = (w_\cap  + 2c\kappa_0\sqrt
	n\frac{\delta}{\sigma})^2 \lesssim w^2_\cap  + 
	n\frac{\delta^2}{\sigma^2}. 
	$$
	Moreover, from the value of $\epsilon = \epsilon(\epsilon_0)$ set above,
	$\frac{2r(1+\epsilon)^2}{\delta\epsilon^2} \leq 2\frac{r}{n\delta}(w^2_\cap + n\frac{\delta^2}{\sigma^2})$, so that \eqref{eq:cond-on-m} is satisfied if \eqref{eq:cond-m} holds. This gives finally that $p_\delta \geq 1 -\eta$ under this condition.
\end{proof}
Interestingly, up to diverging $\log$ factors (possibly due to proof artefacts), the requirement of Prop.~\ref{prop:mainrep} can be seen as a special case of \eqref{eq:cond-m} when $\delta \to 0^+$, \ie for a ``vanishing'' quantiser, since $(w_\cap + \sqrt{2 \log \tinv{\eta}})^2 \lesssim w_{\cap}^2 + \log\tinv{\eta}$. Finally, the application of Prop.~\ref{prop:main} to more than two sets is possible, and will be included in an extended version of this paper. 
\begin{figure*}[t]
	\null\hfill{
		\includegraphics[height=1.9in]{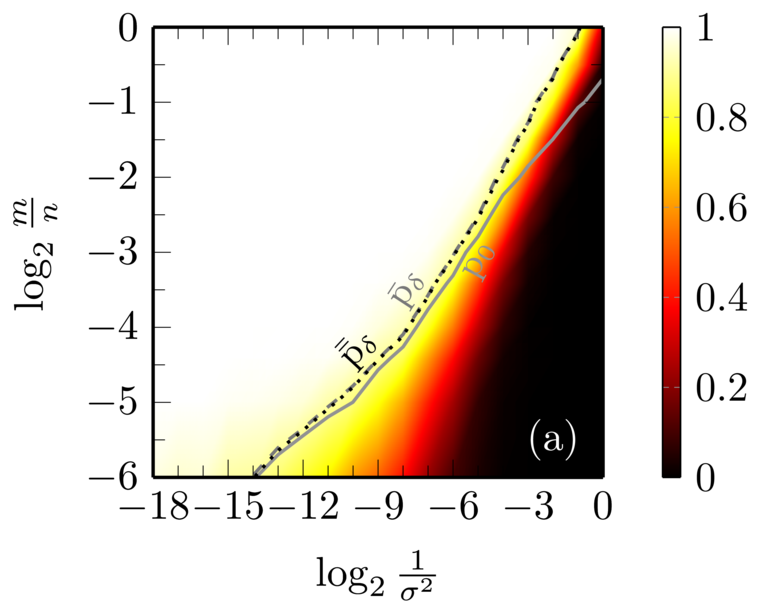}
	}\hfill{
		\includegraphics[height=1.9in]{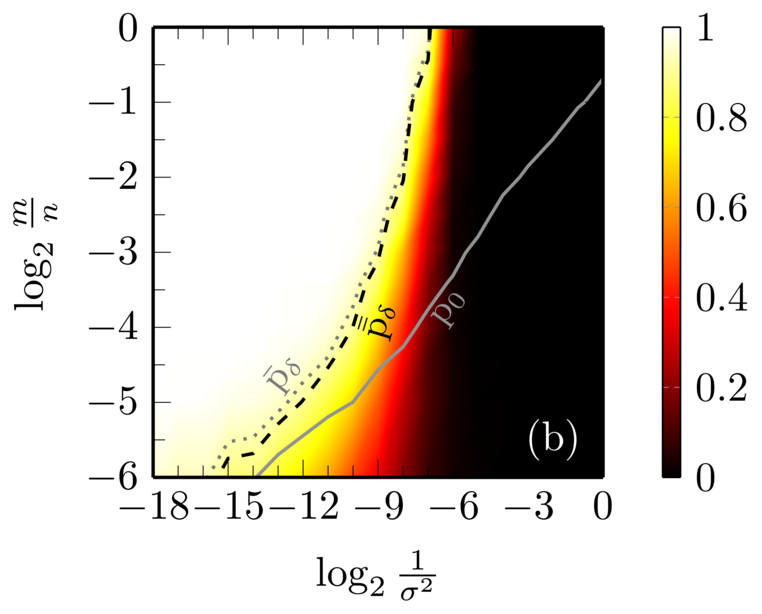}
	}
	\hfill{
		\includegraphics[height=1.9in]{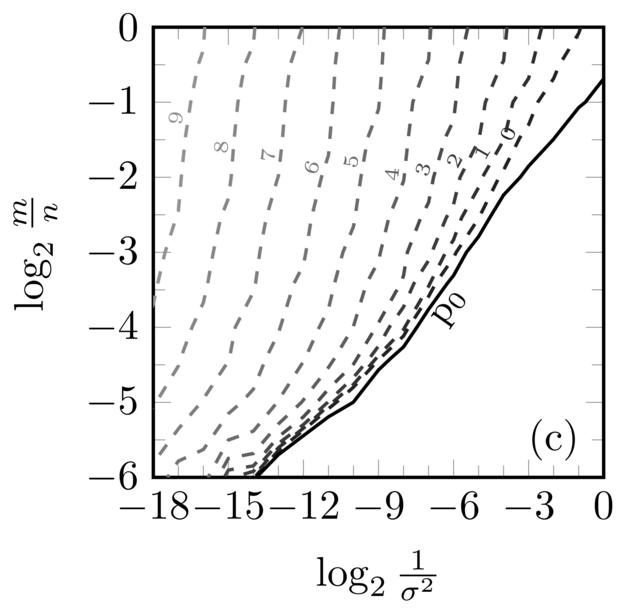}
	}
	\hfill\null
	\vspace{-2ex}
	\caption{\label{fig:empirical-res}Empirical phase transitions of the quantised eclipse problem for the case of two disjoint $\ell_2$-balls; for $128$ random instances of $\bs\Phi$ and as a function of $\frac{1}{\sigma^2}$ and the rate $\frac{m}{n}$, we report for $\delta = 2^1$  {(a)} and $\delta = 2^4$ {(b)} the empirical estimate of $\bar{\bar{p}}_\delta$ (heat map), along with the phase transition curves at probability $0.9$;~(c)~illustrates how the phase transition curves $\bar{\bar{\rm p}}_\delta$  vary for several values of $\delta$, with $\log_2 \delta$ being annotated on each line.
		\vspace{-3ex}}
\end{figure*}
\subsection{Testable Conditions by Convex Problems}
\label{sec:testable-conditions-by-convex-optimisation}
To properly verify the bound on $p_\delta$ in Prop.~\ref{prop:main} we should test the existence of any element in $\sf A(\cl C_1) \cap \sf A(\cl C_2) \subset \delta \bb Z^m$, \ie of any two \emph{consistent} vectors~$\bs x_1 \in \cl C_1, \bs x_2 \in \cl C_2 : \sf A(\bs x_1) = \sf A(\bs x_2)$. As expected, this search is computationally intractable, so we now deduce numerically testable, albeit less tight lower bounds for $p_\delta$. Let us first define the {\em consistency margin} $\tau$ 
\begin{equation}
\label{eq:trueprob}
\tau \coloneqq \min_{\bs z \in \cl C^-} \|\bs\Phi \bs z\|_\infty,
\end{equation} 
that is a function of $\bs\Phi$ and $\cl C^-$, and can be related to the minimal separation $\sigma$ as defined above. Moreover, the event $\tau > \delta$ depends only on $\bs \Phi$, so we can write \eqref{eq:secondformpdelta} as 
\begin{equation}
p_\delta = \bb P_{\bs\Phi,\bs\xi}[\sf E |\tau \leq \delta] \bb P_{\bs\Phi}[\tau \leq\delta] + \bar{p}_{\delta} \geq \bar{p}_{\delta}, \label{eq:boundcomp1}
\end{equation}
where $\bar{p}_\delta \coloneqq \bb P_{\bs\Phi}[\tau > \delta]$ (\ie $\bb P_{\bs\Phi,\bs\xi}[\sf E |\tau > \delta] = 1$) since $\|\bs\Phi (\bs x_1 - \bs x_2)\|_\infty \geq \tau > \delta \implies \sf A(\bs x_1) \neq \sf A(\bs x_2)$, while the converse does not hold. Note that $\bar{p}_\delta$ fully accounts for the cases in which ${\rm Ker}(\bs\Phi) \cap \cl C^- \neq \emptyset$, since if $\bs\Phi \bs z\!=\!\vzer_m\!\implies\!\tau=0$. Clearly, we can now estimate $\bar{p}_{\delta}$ as $\tau$ can be computed for each $\bs \Phi$ when the optimisation problem \eqref{eq:trueprob} is convex (\ie iff $\cl C^-$ is, as for disjoint convex sets).

To tighten this bound and fully leverage dithering, we form a partition $\cup_j \cl C^{(j)} = \cl C^-$ formed by the cones
\[
\cl C^{(j)} \coloneqq \{\bs z \in \cl C^- : \vert \bs \varphi_j^\top \bs z\vert \geq \vert \bs \varphi_{i}^\top \bs z\vert, \forall i \neq j \in [m]\}\subset \cl C^-.
\]
We can now define for $j\in[m], \tau_j \coloneqq \min_{\bs z \in \cl C^{(j)}} \vert\bs\varphi^\top_j \bs z\vert$, where clearly $\tau_j \geq \tau$. 
Letting $\sf A_i(\bs x) \coloneqq \cl Q_\delta(\bs\varphi_i^\top \bs z + \xi_i)$, we use a shorthand for the event $\sf E_i\coloneqq \{\sf A_i(\bs x_1) \neq \sf A_i(\bs x_2)\}$ and bound
\begin{IEEEeqnarray}{rl}
	\nonumber & p_\delta = \bb P_{\bs \Phi,\bs \xi}[\forall \bs x_1 \in \cl C_1, \bs x_2 \in \cl C_2, \exists i \in [m] : \sf E_i] \\
	\nonumber & = \bb P_{\bs \Phi,\bs \xi}[\forall j\!\in\![m],\!\bs x_1\!\in\!\cl C_1,\!\bs x_2\!\in\!\cl C_2,\!\bs x_1\!-\!\bs x_2\!\in\!\cl C^{(j)},\exists i \in [m] : \sf E_i]
	\\
	\nonumber & \geq \bb P_{\bs \Phi,\bs \xi}[\forall j\!\in\![m],\!\bs x_1\!\in\!\cl C_1,\!\bs x_2\!\in\!\cl C_2,\!\bs x_1\!-\!\bs x_2\!\in\!\cl C^{(j)}, \sf E_j]\\
	\nonumber & = \bb E_{\bs \Phi} \bb P_{\bs \xi}[\forall j\!\in\![m],\!\bs x_1\!\in\!\cl C_1,\!\bs x_2\!\in\!\cl C_2,\!\bs x_1\!-\!\bs x_2\!\in\!\cl C^{(j)}, \sf E_j \vert\bs\Phi]
\end{IEEEeqnarray}
Then, since the entries $\xi_j$ of $\bs\xi\!\sim\!\cl U^m([0,\delta])$ are {\em i.i.d.}, 
\begin{IEEEeqnarray}{rl}
	\nonumber & \bb P_{\bs \xi}[\forall j\!\in\![m],\!\bs x_1\!\in\!\cl C_1,\!\bs x_2\!\in\!\cl C_2,\!\bs x_1\!-\!\bs x_2\!\in\!\cl C^{(j)}, \sf E_j\vert\bs\Phi] \\
	\nonumber & = \ts \prod_{j \in [m]}\bb P_{\xi_j}[\forall \bs x_1 \in \cl C_1, \bs x_2 \in \cl C_2: \bs x_1 - \bs x_2 \in \cl C^{(j)}, \sf E_j \vert\bs\Phi] \\
	&=\ts\prod_{j\in[m]}\!\min\{1,\!\tfrac{\tau_j}{\delta}\}\nonumber \\ 
	& \implies\!p_\delta \geq \bb E_{\bs\Phi} \ts \prod_{j\in[m]}\min\{1,\tfrac{\tau_j}{\delta}\} \eqqcolon \bar{\bar{p}}_\delta, 
	\label{eq:boundcomp2}
\end{IEEEeqnarray}
where the second last line follows since $\sf E_j$ occurs whenever, given two intervals $\bs \varphi^\top_j \cl C_1, \bs \varphi^\top_j \cl C_2 \subset \bb R$ that are $\tau_j$ far apart, a quantiser threshold $\delta\bb Z +\xi_j$ falls between them. Hence, this event is identical to having $\bb P[\xi_j \in [0,{\tau_j}]] = \min\{1,\tau_j/\delta\}$ since $\xi_j \sim \cl U([0,\delta])$. The computational complexity of estimating $\bar{\bar{p}}_\delta$ is similar to that of \eqref{eq:boundcomp1}, while \eqref{eq:boundcomp2} is sharper, as it can be shown that $\bar{\bar{p}}_\delta\geq {\bar{p}}_\delta$. However, we expect both bounds to be somewhat loose \emph{w.r.t.} the one in Prop.~\ref{prop:main}.
\subsection{The Case of Two Disjoint $\ell_2$-Balls}\label{sec:the-case-of-two-euclidean-balls}
We now briefly focus on the case of two $\ell_2$-balls $\cl C_1 \coloneqq r_1 \bb B^n_{2} + \bs c_1$ and $\cl C_2 \coloneqq r_2 \bb B^n_{2} + \bs c_2$, for which $\cl C^- =  r \bb B^{n}_{2} + \bs c$ with $\bs c \coloneqq \bs c_1 - \bs c_2$ and $r \coloneqq r_1 + r_2$. It is then shown in \cite[Prop.~4.3]{AmelunxenLotzMcCoyEtAl2014} that ${w_\cap} \lesssim \frac{r}{\|\bs c\|_2} \sqrt n \simeq \frac{\sqrt{n}}{\sigma}$ when $\sigma \gg r$ since $\sigma =  \|\bs c\|_2 - r$. We can now compare the sample complexities in Prop.~\ref{prop:mainrep} and Prop.~\ref{prop:main}: up to some $\log$ and additive factors, we see that Prob.~\ref{prob:qrep} has rate $\tfrac{m}{n} \gtrsim \frac{1}{\sigma^2} (1+ \delta^2)$, while Prob.~\ref{prob:rep} only requires $\tfrac{m}{n} \gtrsim \frac{1}{\sigma^2}$, hence showing the effect of $\delta$ that we will illustrate in our numerical experiments below.

\section{Numerical Experiments}\label{sec:numerical-experiments}
\label{sec:numexps}
We now test the special case of Sec.~\ref{sec:the-case-of-two-euclidean-balls} by generating random instances of $\bs\Phi \sim \cl N^{m\times n}(0,1)$ and\footnote{By uniformity of ${\rm Ker}(\bs\Phi)$, $\bs\Phi \sim \cl N^{m\times n}(0,1)$ over the Grassmannian at the origin, it is legitimate to fix a randomly drawn direction $\tfrac{\bs c}{\|\bs c\|_2}$ for the simulations.} $\cl C^-$, and computing the quantities $\tau_j, j \in [m]$ and $\tau$ for each instance, as specified in Sec.~\ref{sec:testable-conditions-by-convex-optimisation}. This allows us to empirically estimate $\bar{p}_\delta, \bar{\bar{p}}_\delta$ respectively in \eqref{eq:boundcomp1}, \eqref{eq:boundcomp2} on $128$ trials for each of the configurations $n = 2^6$ and $m \in [2^0,2^6]$, and varying  $\cl C^-$ by fixing $r = 2$ and taking $\sigma = \|\bs c\|_2 - r \in [2^0, 2^{9}]$. The estimated values of $\bar{\bar{p}}_\delta$ are then reported as heat maps in Fig. \ref{fig:empirical-res}a,b along with the phase transition curves 
 $\bar{\bar{\rm p}}_\delta \coloneqq \{\bar{\bar{p}}_\delta \geq 0.9\}$, ${\bar{\rm p}}_\delta \coloneqq \{{\bar{p}}_\delta \geq 0.9\}$, and the linear case of Prop.~\ref{prop:mainrep} ${{\rm p}}_0 \coloneqq \{{{p}}_0 \geq 0.9\}$, with $p_0$ being estimated as in \cite{BandeiraMixonRecht2014}. 
Given $\tau_j, j\in[m]$ for all instances, we compute in {Fig.~\ref{fig:empirical-res}c} the phase transition curves corresponding to $\bar{\bar{p}}_\delta$ for several $\delta = \{2^0, 2^1, \ldots, 2^9\}$. For each curve, the event ${\sf A}(\cl C_1) \cap {\sf A}(\cl C_2) = \emptyset$ holds with probability at least $0.9$. These curves are indeed compatible with the fact that $\log_2 \tfrac{m}{n} \gtrsim \log_2 \tfrac{1}{\sigma^2} + \log_2 (1+\delta^2)$ (up to $\log$ factors, and as concluded in Sec.~\ref{sec:the-case-of-two-euclidean-balls}). However, we suspect that $\bar{\bar{p}}_\delta$ is still not sufficiently tight to approach our theoretical, albeit computationally intractable, bound on $p\delta$, and leave this improvement to a future investigation.

\section{Conclusion}
\label{sec:concl-open-quest}
The fundamental limits of learning tasks with embeddings are being tackled in several studies; our result illustrates the requirements for exact classification after quantised random embedding of two disjoint closed convex sets. As we only developed cases in which the datasets $\cl K$ are not specified as low complexity sets, we will discuss them in future works, ~\eg for the case of $Q$ disjoint ``clusters'' of sparse signals $\cl C_i, i \in [Q]$. 

\section{Appendix}
\label{sec:proof-main-prop}

\begin{proof}[Proof of Lemma \ref{lemma:tfd-dithered-embed}] 
	We adapt the proof of~\cite[Prop.~1]{JacquesCambareri2016}. 
	Given $\rho>0$ to be fixed later, let $\cl E_\rho$ be a $\rho$-covering of $\cl E$ in the $\ell_1$-metric, \ie for all $\bs a \in \cl E$ there exists $\bs a_0 \in \cl E_\rho$ such that $\|\bs a - \bs a_0\|_1 \leq \rho$. Notice that since $X(\bs a, \bs b) \coloneqq \cl D_{\ell_1}({\sf A}'(\bs a), {\sf A}'(\bs b)) = \tinv{m}\,\sum_i |\cl Q(a_i + \xi_i) - \cl Q(b_i + \xi_i)| = \tinv{m}\sum_i X_i$, with the \emph{i.i.d.} sub-Gaussian \rvs $X_i$ such that
	$\bb E X_i = |a_i - b_i|$~\cite[App. A]{Jacques2015}, one can easily
	prove the concentration of $X(\bs a, \bs b)$ around $\bb E X(\bs a,
	\bs b) = \cl D_{\ell_1}(\bs a, \bs b)$ both on a fixed pair $\bs a,
	\bs b \in \cl E$ and, by union bound, for all $\bs a, \bs b \in \cl
	E_\rho$ since there are no more than $(e^{\cl H_1(\cl E,\rho)})^2$
	such pairs in $\cl E_\rho$. Unfortunately, the discontinuity of the mapping ${\sf A}'$ prevents us
	to directly extend this over the full set $\cl E$ by a continuity
	argument applied to each neighbourhood of the covering. However, this situation can be overcome by \emph{softening}
	the pseudo-distance $d(\cdot,\cdot) \coloneqq |\cl Q(\cdot) - \cl Q(\cdot)|$ composing $X$~\cite{Jacques2015,PlanVershynin2014}. 
	We first note that 
	$d(a,b) \coloneqq \delta \sum_{k\in \bb Z} \bb I_{\cl S}( a - k\delta, b - k\delta)$,
	where $\cl S = \{(a,b) \in \bb R^2: ab < 0\}$ and $\bb I_{\cl C}(a,b)$
	is the indicator of $\cl C$ evaluated in $(a, b)$, \ie it is equal to $1$ if $(a,b) \in \cl C$ and $0$ otherwise. In fact, $d(a,b) =
	\delta |(\delta \bb Z) \cap [a,b]|$, with $|\cdot|$ the cardinality operator, showing that $d/\delta$
	counts the number of \emph{thresholds} in $\delta \bb Z$ that can be
	inserted between $a$ and $b$. 
	
	Introducing the
	set $\cl S^t = \{(a,b) \in \bb R^2: a < - t, b > t\} \cup \{(a,b) \in
	\bb R^2: a >  t, b < - t\}$ for $t\in \bb R$, with $\cl S^0=\cl S$, we can define a soft
	version of $d$ by
	\begin{equation}
	\label{eq:def-d-t}
	\ts d^t(a,b) \coloneqq \delta \sum_{k\in \bb Z} \bb I_{\cl S^t}( a - k\delta, b
	- k\delta).   
	\end{equation}
	Thanks to $\cl S^t$, the value of $t$ determines a set of forbidden (or relaxed) intervals $\delta \bb Z + [-|t|, |t|] = \{\,[k\delta-|t|, k\delta+|t|]: k \in \bb Z\}$ if $t > 0$ (respectively $t < 0$) of size $2|t|$ and centred on the quantiser thresholds in $\delta \bb Z$. For $t >0$ a threshold of $\delta \bb Z$ is not counted in $d^t(a,b)$ if $a$ or $b$ fall in its forbidden interval, whereas for $t<0$ a threshold that is not between $a$ and $b$ can be counted if $a$ or $b$ fall inside its relaxed interval. 
	
	By extension, we can also define $\cl D^t(\bs a, \bs b) \coloneqq \tinv{m}
	\sum_{i} d^t(a_i,b_i)$ for $\bs a, \bs b \in
	\bb R^m$, so that $\cl D^0(\bs a, \bs b) = \cl D_{\ell_1}\big(\cl Q(\bs a), \cl
	Q(\bs b)\big)$. 	
	Interestingly, this distance displays the following	continuity property~\cite[Lemma 2]{Jacques2015}. For $\bs a,\bs b \in \cl E$, and $\bs a_0, \bs b_0$ their respective closest points in $\cl E_\rho$ we have, for every $t\in \bb R$ and\footnote{In~\cite[Lemma 2]{Jacques2015}, it is assumed $P \geq 1$ but nothing prevents $P>0$.}~$P > 0$,
        \begin{IEEEeqnarray}{C}
          \label{eq:continuity-L1-Dt-left}
          \ts \cl D^t(\bs a,\bs b) \geq \cl D^{t+\tfrac{\rho P}{m}}(\bs a_0,\bs b_0) - 8 (\frac{\delta}{P} + \tfrac{\rho}{m}),\\        
          \label{eq:continuity-L1-Dt-right}
          \ts \cl D^t(\bs a,\bs b) \leq \cl D^{t-\tfrac{\rho P}{m}}(\bs a_0,\bs b_0) + 8 (\frac{\delta}{P}+ \tfrac{\rho}{m}).
        \end{IEEEeqnarray}
        Moreover, for $\bs \xi \sim \cl U^m([0,\delta])$ and $\bs a, \bs b$ fixed, $\cl D^t(\bs a + \bs \xi, \bs b + \bs \xi)$ concentrates around its mean which is close to $\cl D_{\ell_1}(\bs a, \bs b)$~\cite[Lemma 3]{Jacques2015}. In fact, $\bb |\bb E \cl D^t(\bs a + \bs \xi, \bs b + \bs \xi) - \cl D_{\ell_1}(\bs a, \bs b) | \lesssim |t|$, so that for some $c>0$,
	$$
	\ts \bb P\big[|\cl D^t(\bs a + \bs \xi, \bs b + \bs \xi) - \cl D_{\ell_1}(\bs a, \bs b)| > 4|t| + \epsilon (\delta + |t|)\big]
	\lesssim e^{-c \epsilon^2 m}.
	$$
	Therefore, by union bound and for some $P > 0$ to be fixed soon, if $m \gtrsim \epsilon^{-2} \cl H_1(\cl E,
	\rho)$ then
	\begin{align}
	&\ts| \cl D^{t\pm\frac{\rho P}{m}} (\bs a_0' + \bs \xi, \bs b_0' + \bs \xi) -
	\cl D_{\ell_1}(\bs a_0', \bs b_0') \big |\nonumber\\
	&\ts\label{eq:concent-cover-upper-2}
	\qquad\leq
	4|t|+4 \frac{\rho P}{m} +\epsilon(\delta+ |t| + \frac{\rho P}{m}),\ \forall \bs a_0',\bs b_0' \in \cl E_\rho,
	\end{align}
	with probability exceeding $1 - C e^{-c \epsilon^2 m}$ for some
	$C,c>0$. 
	
	Consequently, for any $\bs a, \bs b \in \cl E$ and $\bs a_0, \bs b_0$ their respective closest point in $\cl E_\rho$, using
	\eqref{eq:continuity-L1-Dt-right} combined with \eqref{eq:concent-cover-upper-2}, and since the triangular inequality provides $\cl D_{\ell_1}(\bs a_0, \bs b_0) \leq \cl D_{\ell_1}(\bs a, \bs b) + \frac{2\rho}{m}$, we have with the same probability and
	for some $c>0$,\vspace{-1ex}
	\begin{align*}
	&\ts\! \cl D^{t} (\bs a + \bs \xi, \bs b + \bs \xi) \leq \cl D^{t-\tfrac{\rho P}{m}} (\bs a_0 + \bs \xi, \bs b_0 + \bs \xi)   
	+ 8 (\frac{\delta}{P} + \tfrac{\rho}{m})\\
	&\ts \!\!\leq \cl D_{\ell_1}(\bs a_0, \bs b_0) + 
	4|t|+ \frac{4\rho P}{m} +\epsilon(\delta+ |t| +\tfrac{\rho P}{m}) + 8 (\frac{\delta}{P}  + \tfrac{\rho}{m})\\
	&\ts \!\! \leq \cl D_{\ell_1}(\bs a, \bs b) + c\big(|t|+ \frac{\rho (1 + P(1+\epsilon))}{m} +\epsilon(\delta+ |t|) + (\frac{\delta}{P}  + \tfrac{\rho}{m})\big)\\
	&\leq \cl D_{\ell_1}(\bs a, \bs b) + c(|t| +
	\delta\epsilon),
	\end{align*}
	where we finally set the free parameters as
	$P^{-1} = \epsilon$ and $\rho = m\delta \tfrac{\epsilon^{2}}{1+\epsilon} < m\delta\min(\epsilon,\epsilon^2)$, giving $\rho P \leq
	m\delta \epsilon$ and $\tfrac{\rho}{m} \leq \delta \epsilon$. The lower bound is obtained
	similarly using \eqref{eq:continuity-L1-Dt-left} with the minus case of~\eqref{eq:concent-cover-upper-2}, and
	Prop.~\ref{lemma:tfd-dithered-embed} is finally obtained with $t=0$.
\end{proof}
\vtodo{Remarks: 
	* Note that in the small sigma regime, Prop. 1 diverge on the requirement on 
$m$ (except if we keep $\delta < C sigma$). I guess this make sense somehow and illustrate the special quantised geometry compared to linear REP that doesn't display such a divergence.
}
\bibliographystyle{IEEEtran}
\bibliography{qcs-qrep}

\begin{thebibliography}{10}
\providecommand{\url}[1]{#1}
\csname url@samestyle\endcsname
\providecommand{\newblock}{\relax}
\providecommand{\bibinfo}[2]{#2}
\providecommand{\BIBentrySTDinterwordspacing}{\spaceskip=0pt\relax}
\providecommand{\BIBentryALTinterwordstretchfactor}{4}
\providecommand{\BIBentryALTinterwordspacing}{\spaceskip=\fontdimen2\font plus
\BIBentryALTinterwordstretchfactor\fontdimen3\font minus
  \fontdimen4\font\relax}
\providecommand{\BIBforeignlanguage}[2]{{%
\expandafter\ifx\csname l@#1\endcsname\relax
\typeout{** WARNING: IEEEtran.bst: No hyphenation pattern has been}%
\typeout{** loaded for the language `#1'. Using the pattern for}%
\typeout{** the default language instead.}%
\else
\language=\csname l@#1\endcsname
\fi
#2}}
\providecommand{\BIBdecl}{\relax}
\BIBdecl

\bibitem{GrayNeuhoff1998}
R.~M. Gray and D.~L. Neuhoff, ``{Q}uantization,'' \emph{{IEEE} {T}ransactions
  on {I}nformation {T}heory}, vol.~44, no.~6, pp. 2325--2383, 1998.

\bibitem{Jacques2015}
L.~Jacques, ``{S}mall width, low distortions: quasi-isometric embeddings with
  quantized sub-{Gauss}ian random projections,'' \emph{ar{X}iv preprint
  ar{X}iv:1504.06170}, 2015.

\bibitem{BoufounosJacquesKrahmerEtAl2015}
P.~T. Boufounos, L.~Jacques, F.~Krahmer, and R.~Saab, ``{Q}uantization and
  compressive sensing,'' in \emph{Compressed Sensing and its
  Applications}.\hskip 1em plus 0.5em minus 0.4em\relax Springer, 2015, pp.
  193--237.

\bibitem{RahimiRecht2008}
A.~Rahimi and B.~Recht, ``{R}andom {F}eatures for {L}arge-{S}cale {K}ernel
  {M}achines,'' in \emph{Advances in Neural Information Processing Systems 20},
  J.~C. Platt, D.~Koller, Y.~Singer, and S.~T. Roweis, Eds.\hskip 1em plus
  0.5em minus 0.4em\relax Curran Associates, Inc., 2008, pp. 1177--1184.

\bibitem{BoufounosRaneMansour2015}
P.~T. Boufounos, S.~Rane, and H.~Mansour, ``{R}epresentation and {C}oding of
  {S}ignal {G}eometry,'' \emph{ar{X}iv preprint ar{X}iv:1512.07636}, 2015.

\bibitem{JacquesCambareri2016}
L.~Jacques and V.~Cambareri, ``{T}ime for dithering: fast and quantized random
  embeddings via the restricted isometry property,'' \emph{ar{X}iv preprint
  ar{X}iv:1607.00816}, 2016.

\bibitem{MoshtaghpourJacquesCambareriEtAl2016}
A.~Moshtaghpour, L.~Jacques, V.~Cambareri, K.~Degraux, and C.~De~Vleeschouwer,
  ``{C}onsistent {B}asis {P}ursuit for {S}ignal and {M}atrix {E}stimates in
  {Q}uantized {C}ompressed {S}ensing,'' \emph{{IEEE} {S}ignal {P}rocessing
  {L}etters}, vol.~23, no.~1, pp. 25--29, 2016.

\bibitem{BandeiraMixonRecht2014}
A.~S. Bandeira, D.~G. Mixon, and B.~Recht, ``{C}ompressive classification and
  the rare eclipse problem,'' \emph{ar{X}iv preprint ar{X}iv:1404.3203}, 2014.

\bibitem{Vershynin2012}
R.~Vershynin, ``{I}ntroduction to the non-asymptotic analysis of random
  matrices,'' in \emph{Compressed Sensing: Theory and Applications}.\hskip 1em
  plus 0.5em minus 0.4em\relax Cambridge University Press, 2012, pp. 210--268.

\bibitem{JohnsonLindenstrauss1984}
W.~B. Johnson and J.~Lindenstrauss, ``{E}xtensions of {L}ipschitz mappings into
  a {Hilbert} space,'' \emph{{C}ontemporary {M}athematics}, vol.~26, no.
  189-206, pp. 1--1, 1984.

\bibitem{DasguptaGupta2003}
S.~Dasgupta and A.~Gupta, ``\BIBforeignlanguage{en}{{A}n elementary proof of a
  theorem of {Johnson} and {Lindenstrauss}},''
  \emph{\BIBforeignlanguage{en}{{R}andom {S}tructures \& {A}lgorithms}},
  vol.~22, no.~1, pp. 60--65, Jan. 2003.

\bibitem{Achlioptas2003}
D.~Achlioptas, ``{D}atabase-friendly random projections:
  {Johnson-Lindenstrauss} with binary coins,'' \emph{{J}ournal of {C}omputer
  and {S}ystem {S}ciences}, vol.~66, no.~4, pp. 671--687, Jun. 2003.

\bibitem{OymakRecht2015}
S.~Oymak and B.~Recht, ``{N}ear-{O}ptimal {B}ounds for {B}inary {E}mbeddings of
  {A}rbitrary {S}ets,'' \emph{ar{X}iv preprint ar{X}iv:1512.04433}, 2015.

\bibitem{Jacques2015a}
L.~Jacques, ``{A Quantized Johnson--Lindenstrauss lemma: The finding of
  Buffon{\textquoteright}s needle},'' \emph{{IEEE} {T}ransactions on
  {I}nformation {T}heory}, vol.~61, no.~9, pp. 5012--5027, 2015.

\bibitem{PlanVershynin2014}
Y.~Plan and R.~Vershynin, ``{D}imension reduction by random hyperplane
  tessellations,'' \emph{{D}iscrete \& {C}omputational {G}eometry}, vol.~51,
  no.~2, pp. 438--461, 2014.

\bibitem{PlanVershynin2013}
------, ``{R}obust 1\mbox{-}bit compressed sensing and sparse logistic
  regression: {A} convex programming approach,'' \emph{{IEEE} {T}ransactions on
  {I}nformation {T}heory}, vol.~59, no.~1, pp. 482--494, 2013.

\bibitem{Dasgupta1999}
S.~Dasgupta, ``{L}earning mixtures of {Gaussians},'' in
  \emph{40\textsuperscript{th} {Annual} {Symposium} on {Foundations} of
  {Computer} {Science}}, 1999, pp. 634--644.

\bibitem{DavenportBoufounosWakinEtAl2010}
M.~A. Davenport, P.~T. Boufounos, M.~B. Wakin, and R.~G. Baraniuk, ``{S}ignal
  processing with compressive measurements,'' \emph{{IEEE} {J}ournal of
  {S}elected {T}opics in {S}ignal {P}rocessing}, vol.~4, no.~2, pp. 445--460,
  2010.

\bibitem{ReboredoRennaCalderbankEtAl2013}
H.~Reboredo, F.~Renna, R.~Calderbank, and M.~Rodrigues, ``{C}ompressive
  classification,'' in \emph{2013 {IEEE} {International} {Symposium} on
  {Information} {Theory} Proceedings ({ISIT})}, Jul. 2013, pp. 674--678.

\bibitem{ReboredoRennaCalderbankEtAl2016}
H.~Reboredo, F.~Renna, R.~Calderbank, and M.~R. Rodrigues, ``{B}ounds on the
  {N}umber of {M}easurements for {R}eliable {C}ompressive {C}lassification,''
  \emph{{IEEE} {T}ransactions on {S}ignal {P}rocessing}, vol.~64, no.~22, pp.
  5778--5793, 2016.

\bibitem{Gordon1988}
Y.~Gordon, ``{O}n {M}ilman's inequality and random subspaces which escape
  through a mesh in $\mathbb{R}$ n,'' in \emph{Geometric Aspects of Functional
  Analysis}.\hskip 1em plus 0.5em minus 0.4em\relax Springer, 1988, pp.
  84--106.

\bibitem{Schechtman2006}
G.~Schechtman, ``{Two observations regarding embedding subsets of {Euclid}ean
  spaces in normed spaces},'' \emph{{A}dvances in {M}athematics}, vol. 200,
  no.~1, pp. 125--135, 2006.

\bibitem{Pisier1999}
G.~Pisier, \emph{{T}he volume of convex bodies and {B}anach space
  geometry}.\hskip 1em plus 0.5em minus 0.4em\relax Cambridge University Press,
  1999, vol.~94.

\bibitem{AmelunxenLotzMcCoyEtAl2014}
D.~Amelunxen, M.~Lotz, M.~B. McCoy, and J.~A. Tropp, ``{L}iving on the edge:
  phase transitions in convex programs with random data,'' \emph{{I}nformation
  and {I}nference: {A} {J}ournal of the {IMA}}, vol.~3, no.~3, p. 224, 2014.

\end{thebibliography}

\end{document}